\documentclass[a4paper,UKenglish,cleveref, autoref, thm-restate]{lipics-v2021}

\newcommand{\domino}[1]{#1\texttt{-Domino}}
\renewcommand{\phi}{\varphi}
\newcommand{\manyone}{\preccurlyeq_m}
\newcommand{\coRE}{\Pi_1^0}
\newcommand{\shapeset}{\mathcal{T}}
\newcommand{\tileset}{\mathbf{T}}
\newcommand{\reduction}{\phi}
\newcommand{\blank}{\mathrm{blank}}
\newcommand{\dom}{\textrm{dom}}

\title{The Domino problem is undecidable on every rhombus subshift} 

\author{Benjamin {Hellouin de Menibus}}{Université Publique \and Université Paris-Saclay, CNRS, Laboratoire Interdisciplinaire des Sciences du Numérique, 91400, Orsay, France \and \url{https://www.lisn.fr/~hellouin} }{hellouin@lisn.fr}{https://orcid.org/0000-0001-5194-929X}{}

\author{Victor H. {Lutfalla}}{Université Publique \and GREYC, Univ. Caen \and \url{https://www.lutfalla.fr}}{victor@lutfalla.fr}{https://orcid.org/0000-0002-1261-0661}{}

\author{Camille Noûs}{Université Publique}{}{}{}

\authorrunning{B. Hellouin de Menibus, V. Lutfalla and C. Noûs} 

\Copyright{B. Hellouin de Menibus, V. Lutfalla and C. Noûs}

\ccsdesc[500]{Mathematics of computing~Discrete mathematics}
\ccsdesc[500]{Theory of computation~Computability}

\keywords{Rhombus tiling, decidability, Domino problem} 

\funding{ANR C\_SyDySi}
\nolinenumbers 
\hideLIPIcs 

\EventEditors{John Q. Open and Joan R. Access}
\EventNoEds{2}
\EventLongTitle{42nd Conference on Very Important Topics (CVIT 2016)}
\EventShortTitle{CVIT 2016}
\EventAcronym{CVIT}
\EventYear{2016}
\EventDate{December 24--27, 2016}
\EventLocation{Little Whinging, United Kingdom}
\EventLogo{}
\SeriesVolume{42}
\ArticleNo{23}

\begin{document}

\maketitle 

\begin{abstract} 
We extend the classical Domino problem to any tiling of rhombus-shaped tiles.
For any subshift $X$ of edge-to-edge rhombus tilings, such as the Penrose subshift, we prove that the associated $X$-Domino problem is $\coRE$-hard and therefore undecidable. 
It is $\coRE$-complete when the subshift $X$ is given by a computable sequence of forbidden patterns. \end{abstract}

\section{Introduction}
Tilings come, roughly speaking, in two families. Geometrical tilings are coverings of a space, usually the euclidean plane, by geometrical tiles without overlap; the constraints come from the geometry of the tiles. A famous example of geometrical tilings is the Penrose tilings \cite{penrose1974}. Symbolic tilings are colourings of a discrete structure, usually $\mathbb{Z}^d$ for some $d$, whose constraints are given by forbidden patterns. 

Both families have received a lot of attention for their dynamical and combinatorials properties. 
The specificities of geometrical tilings are their symmetries and their links with mathematical cristallography \cite{baake2013}. 
Symbolic tilings, on the other hand, have more links to computability and decidability theory. The seminal example is the Domino problem: given a set of colours and a set of forbidden patterns, is there a colouring of $\mathbb Z^2$ that satisfies those constraints? The proof by Berger \cite{berger1966} that this problem is undecidable shaped the whole domain of research.

There has been much work to extend Domino problems to structures that extend the classical symbolic results on $\mathbb Z^2$. An active area of research considers Domino problem on groups in order to relate properties of the group and of the tiling spaces; see \cite{aubrun2018} for a survey. Other considered extensions are Domino problems in self-similar structures to understand the limit between dimension 1 and 2 \cite{barbieri2016} or Domino problems inside a $\mathbb Z^2$-subshift to understand the effect of dynamical restrictions \cite{aubrun2020}.

Coming back to geometrical tilings, complex examples such as the Penrose tilings were originally defined with jigsaw-type tiles with indentations on their edges (see Fig.~\ref{fig:penrose_tiles}). It can be restated as simple polygon tiles with symbols on their edges \cite{debruijn1981} with the condition that symbols must match. 
In essence these tilings are both symbolic and geometrical.

\begin{figure}[t]
	
	\begin{subfigure}{0.48\textwidth}
		\center\includegraphics[width=0.6\textwidth]{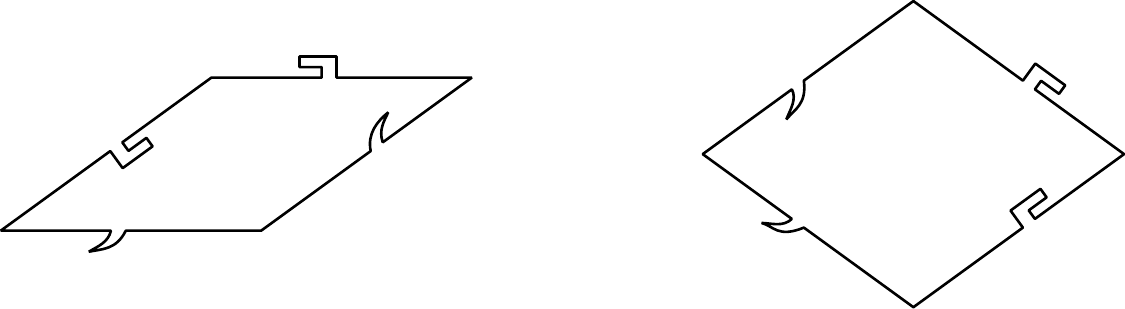}
		\caption{Original definition of the Penrose tiles \cite{penrose1974} with cuts and notches.}
	\end{subfigure}
	\hfill
	\begin{subfigure}{0.48\textwidth}
		\center\includegraphics[width=0.6\textwidth]{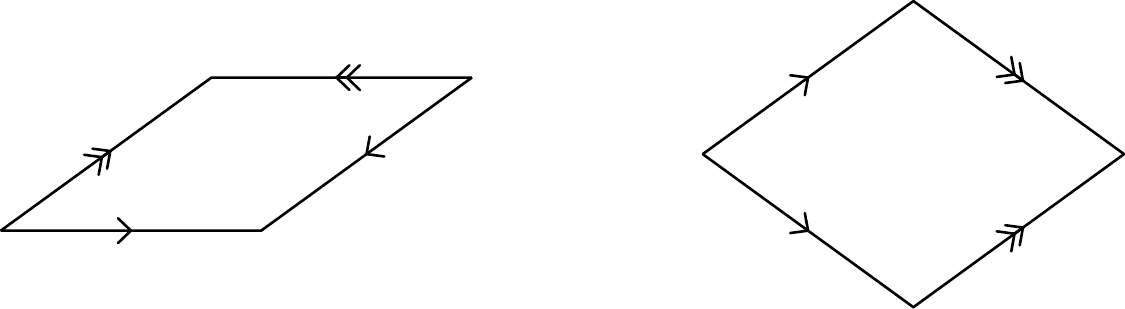}
		\caption{Alternative definition with arrows on the edges \cite{debruijn1981}. 
		The type and direction of the arrow must match between adjacent tiles.
		}
	\end{subfigure}
	\caption{Definitions of Penrose tiles: adding symbols on the edges of the tiles allow for simple rhombus shapes.}
	\label{fig:penrose_tiles}
\end{figure}
Similarly, given a set of shapes, we define symbolic tiles on those shapes by adding colours on the edges and study the induced symbolic-geometrical tilings. 
Let us consider the set of symbolic-geometrical tiles $\tileset$ of Fig.~\ref{fig:example_penrose_domino}. Since the shapes are Penrose rhombuses, there are two natural questions regarding this tileset:
\begin{enumerate}
	\item is there an infinite valid tiling of $\mathbb{R}^2$ with tiles in $\tileset$ up to translation? 
	\item is there an infinite valid tiling of $\mathbb{R}^2$ with tiles in $\tileset$ up to translation that projects to a geometrical Penrose tiling (that is, when removing the coulours from the tiles)?
\end{enumerate}

\begin{figure}[b]
	\center\includegraphics[width=0.8\textwidth]{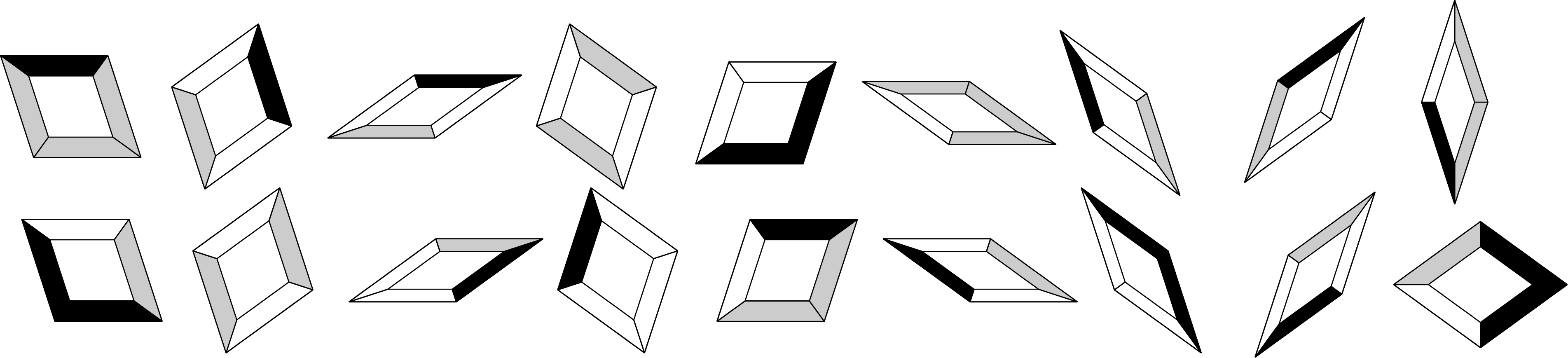}	
	\caption{An example of 
	symbolic-geometrical tiles on Penrose rhombuses.}
	\label{fig:example_penrose_domino}
\end{figure}
It is not hard to see that the first question is at least as hard as (and is in fact equivalent to) the classical Domino problem, which corresponds to the case where the input tiles are all the same rhombus. This motivates us to study the second question, where the geometrical (Penrose) subshift forces the use of the diferent rhombuses.

These are instances of the Domino problem in geometrical subshifts, which is the object of the present article. For any set of rhombuses $\shapeset$ and a geometrical subshift $X$ on these rhombuses, the $\domino{X}$ problem is defined as follows: given as input a set of tiles, that is, rhombuses from $\shapeset$ with a colour on each edge, decide whether it is possible to tile the plane in such a way that:
\begin{enumerate}
	\item tiles with a common edge have the same colour along the shared edge, and
	\item the geometrical tiling (when colours are erased) is valid for $X$.
\end{enumerate}
Our main result is the following:
\begin{theorem}
	Let $X$ be a geometrical subshift given by a computable list of forbidden patterns. $\domino{X}$ is many-one equivalent to the classical Domino problem on $\mathbb Z^2$, that is, co-computably enumerable-complete, and thus undecidable.
\end{theorem}

\section{Definitions}
\subsection{Geometrical tiling spaces}
\begin{definition}[Shapes and patches]
	We call \emph{shape} a geometrical rhombus given as a pair of vectors $(\vec{u},\vec{v})$ and a position $p$.\\
	We call \emph{shapeset} a finite set of shapes considered up-to-translation, see Fig.~\ref{fig:shapeset}.\\
	We call \emph{patch} an edge-to-edge simply connected finite set of shapes, \emph{i.e.}, any two tiles are either disjoint, share a single common vertex or a full common edge, and there is no hole in the patch. We call \emph{support} of a patch the union of its shapes.
	We call \emph{pattern} a patch up to translation.
\end{definition}
Note that shapes are not taken up to rotation. 
\begin{definition}[Tilings, full shift and subshifts]
	Given a shapeset $\shapeset$, we call \emph{$\shapeset$-tiling} an edge-to-edge covering of the euclidean plane without overlap by translates of the shapes in $\shapeset$: see Fig.~\ref{fig:geometrical_tilings}.
	
	We call \emph{full shift} on $\shapeset$, denoted by $X_\shapeset$, the set of all $\shapeset$-tilings.
	
	We call \emph{subshift} of $X_\shapeset$ any subset $X$ of $X_\shapeset$ that is invariant by translation and closed for the tiling topology \cite{robinson2004}.
\end{definition}

\begin{figure}[!b]
	\begin{subfigure}{\textwidth}
		\includegraphics[width=\textwidth]{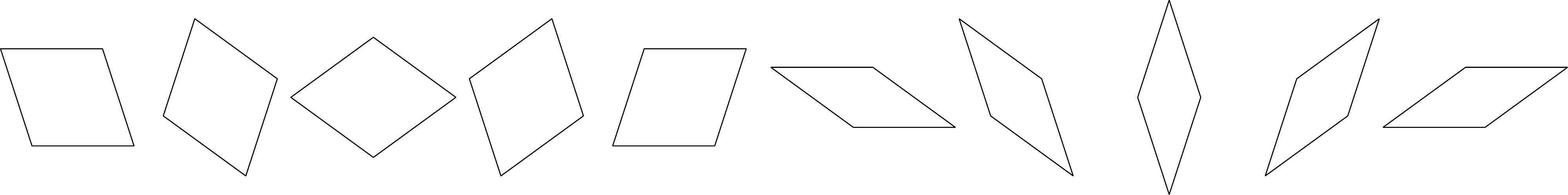}
		\caption{An example of shapeset: the set $\shapeset_{pen}$ of Penrose rhombuses up to translation.}
		\label{fig:shapeset}
	\end{subfigure}
	\smallskip
	
	\begin{subfigure}{\textwidth}
		\includegraphics[width=0.45\textwidth]{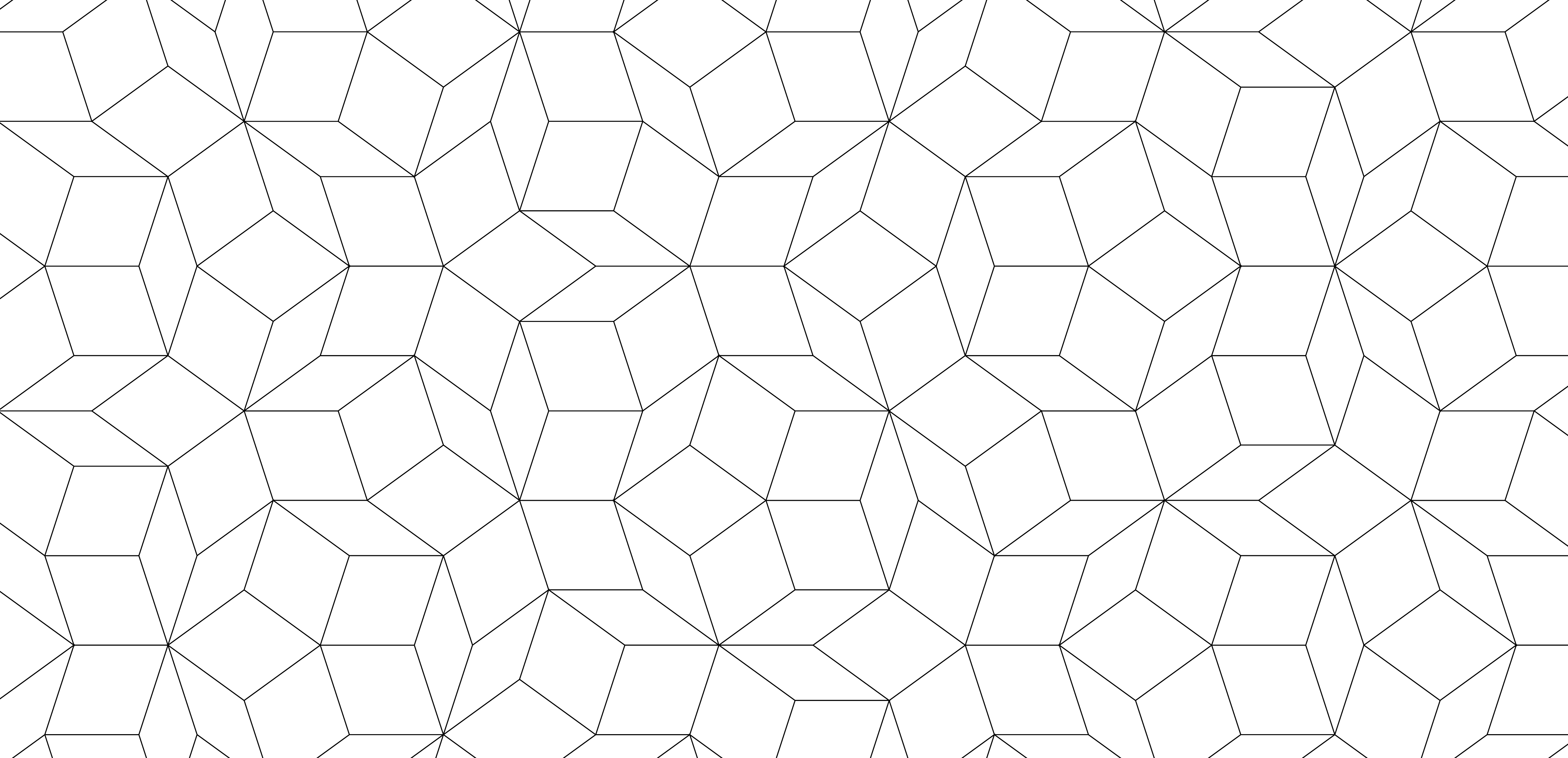}\hfill
		\includegraphics[width=0.45\textwidth]{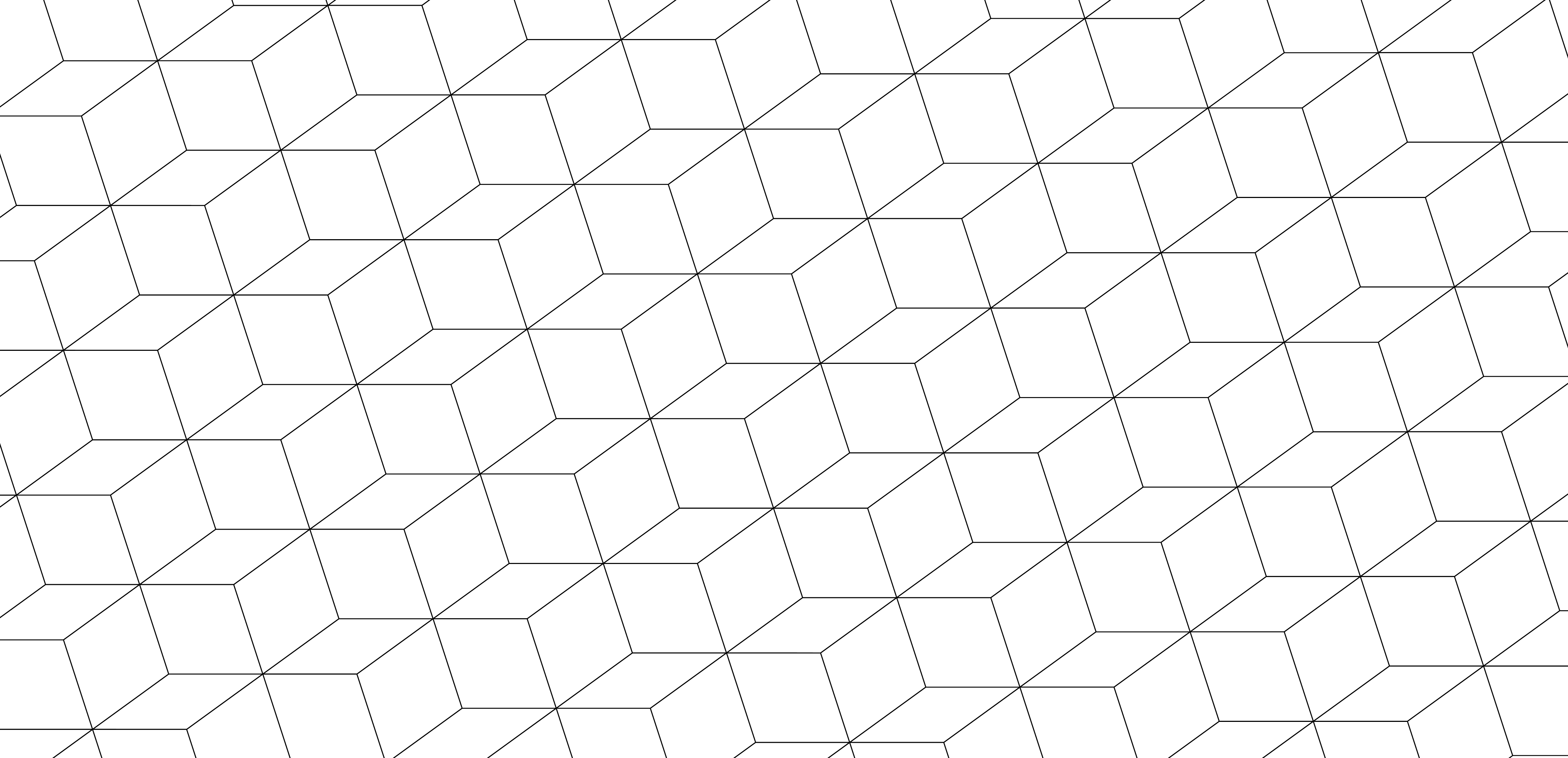}
		\caption{Two examples of geometrical $\shapeset_{pen}$-tilings: on the left a geometrical Penrose tiling and on the right a "cube"-tiling.}
		\label{fig:geometrical_tilings}
	\end{subfigure}
	\caption{Shapes and geometrical tilings.}
	\label{fig:geometrical}
\end{figure}

Edge-to-edge rhombus tilings with finitely many shapes up to translation have \emph{Finite Local Complexity} (FLC): that is, for any compact $C\subseteq \mathbb{R}^2$, there are finitely many patterns whose support is included in $C$.
The FLC hypothesis appears a lot in the study of geometrical tilings and Delone sets; see \cite{baake2013}. In particular, FLC ensures that the tiling space shares most properties with standard $\mathbb{Z}^d$ tiling spaces, such as being compact for the usual tiling topology \cite{robinson2004,lutfalla2022}.

A subshift $X$ can always be characterized by a countable (possibly infinite) set of forbidden patterns, that we denote as a sequence $\mathcal{F}:= (f_n)_{n\in\mathbb{N}}$. 
In other words, $X$ is the set of all tilings where no pattern in $\mathcal F$ appears.
When $\mathcal F$ is computable, we say that the subshift $X$ is \emph{effective}.

We say that a pattern (or patch) has \emph{minimal radius $r$} when its support contains a disk of radius $r$ centered on a vertex of the pattern, and when removing any shape from the patch would break that property. 

A sequence $\mathcal{F}$ of forbidden patterns being fixed, we call locally-allowed patterns $\mathcal{A}(X)$ the set of finite patterns where no forbidden pattern appears, 
and rank-$r$ locally-allowed patterns $\mathcal{A}_r(X)$ the set of patterns of minimal radius $r$ that do not contain any of the first $r$ forbidden patterns. 
Note that when $X$ is an FLC subshift $\mathcal{A}_r(X)$ is finite for all $r$: indeed, there exists a constant $d$ (maximum diameter of the shapes) such that the support of any minimal radius $r$ pattern is included in an $r+d$ disk.

Note that patterns in $\mathcal A(X)$ may not be globally allowed in $X$ (appear in no infinite tiling in $X$).
They may even appear in no tiling of the full shift $X_\shapeset$ if they contain a geometrical impossibility.
Such patterns are called \emph{deceptions} \cite{dworkin1995}.

The interest of locally allowed patterns is that, as seen below in Lemma \ref{lemma:locally_allowed}, they are computable from the list of forbidden patterns, whereas it is not the case for globally allowed patterns.

\subsection{Symbolic-geometrical tiling spaces}
\begin{definition}[Symbols, tiles and tilesets]
	A \emph{tile} is a shape endowed with a colour on each edge, as seen in Fig.~\ref{fig:wang_rhombus}(a). 
	
	Formally, given a finite set $C$ whose elements are called colours and a rhombus shape $r$, we call \emph{$r$-Wang tile} or simply \emph{$r$-tile} a quintuple $(r, a_0, a_1, a_2, a_3)$ with $a_i \in C$. 
	Formally, with $r = (\vec{u}, \vec{v}, p)$, the side $(p, p+\vec{u})$ has colour $a_0$, the side $(p+\vec{u}, p+\vec{u}+\vec{v})$ has colour $a_1$ and so on.

	Given a shapeset $\shapeset$, we call \emph{$\shapeset$-tile} a $r$-tile for some $r\in \shapeset$.
	We call \emph{$\shapeset$-tileset} a finite set of $\shapeset$-tiles, considered up to translation, such that each shape has at least a tile.	
\end{definition}
\begin{definition}[Colour erasing operator $\pi$]
	We define the colour erasing operator $\pi$ by: 
	\begin{itemize}
		\item for any $r$-tile $t$, $\pi(t) := r$
		\item for a tiling $x$ (or finite patch of tiles), $\pi(x) := \{ \pi(t),\ t \in x\}$
		\item for a set of tilings $X$, $\pi(X) := \{ \pi(x),\ x \in X\}$
	\end{itemize}
\end{definition}
\begin{definition}[Tiling]
	Given a finite set of colours $C$ and a $\shapeset$-tileset $\tileset$, we call $\tileset$-tiling a tiling $x$ such that $\pi(x) \in X_\shapeset$ and such that any two tiles in $x$ that share an edge have the same colour on their shared edge. See Fig.~\ref{fig:wang_rhombus}(b).
	
	We denote by $X_\tileset$ the subshift of all $\tileset$-tilings. 
\end{definition}
\begin{figure}[t]
	
	\begin{subfigure}{0.48\textwidth}
		\includegraphics[width=\textwidth]{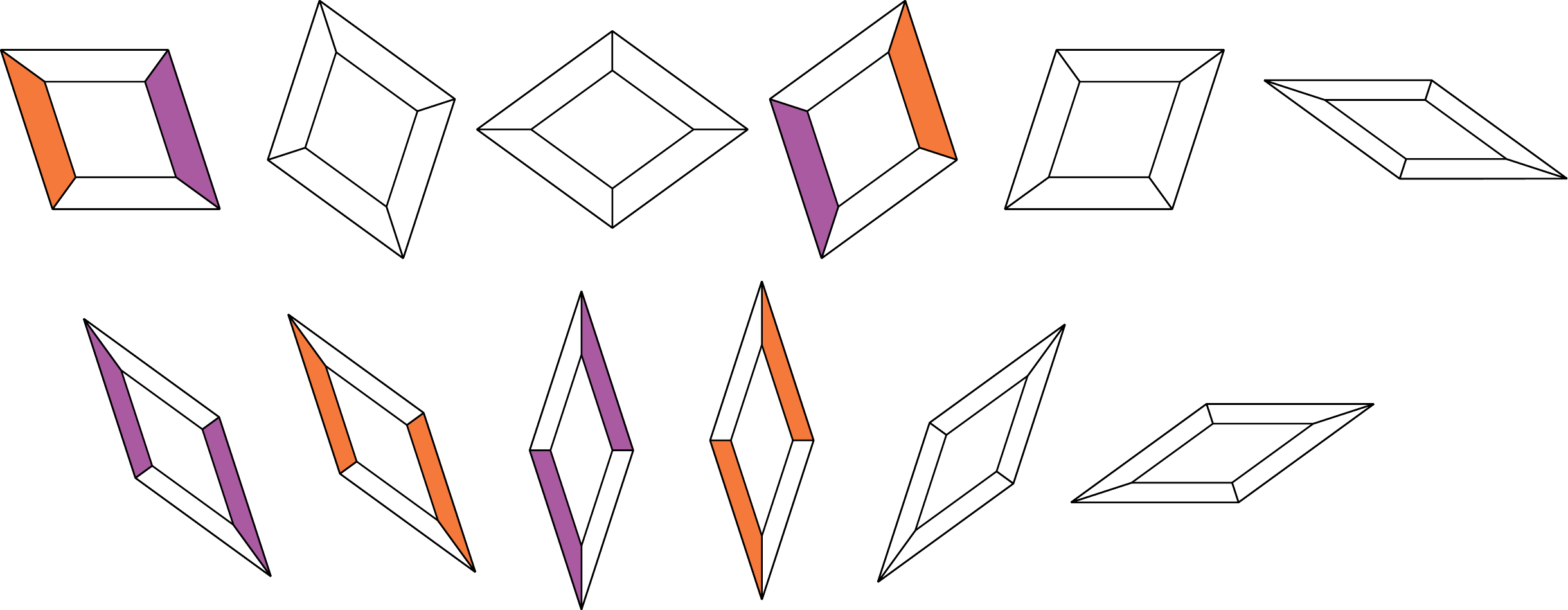}
		\caption{An example of $\shapeset_{pen}$-tileset.}
	\end{subfigure}
	\hfill
	\begin{subfigure}{0.48\textwidth}
		\center
		\includegraphics[width=0.9\textwidth]{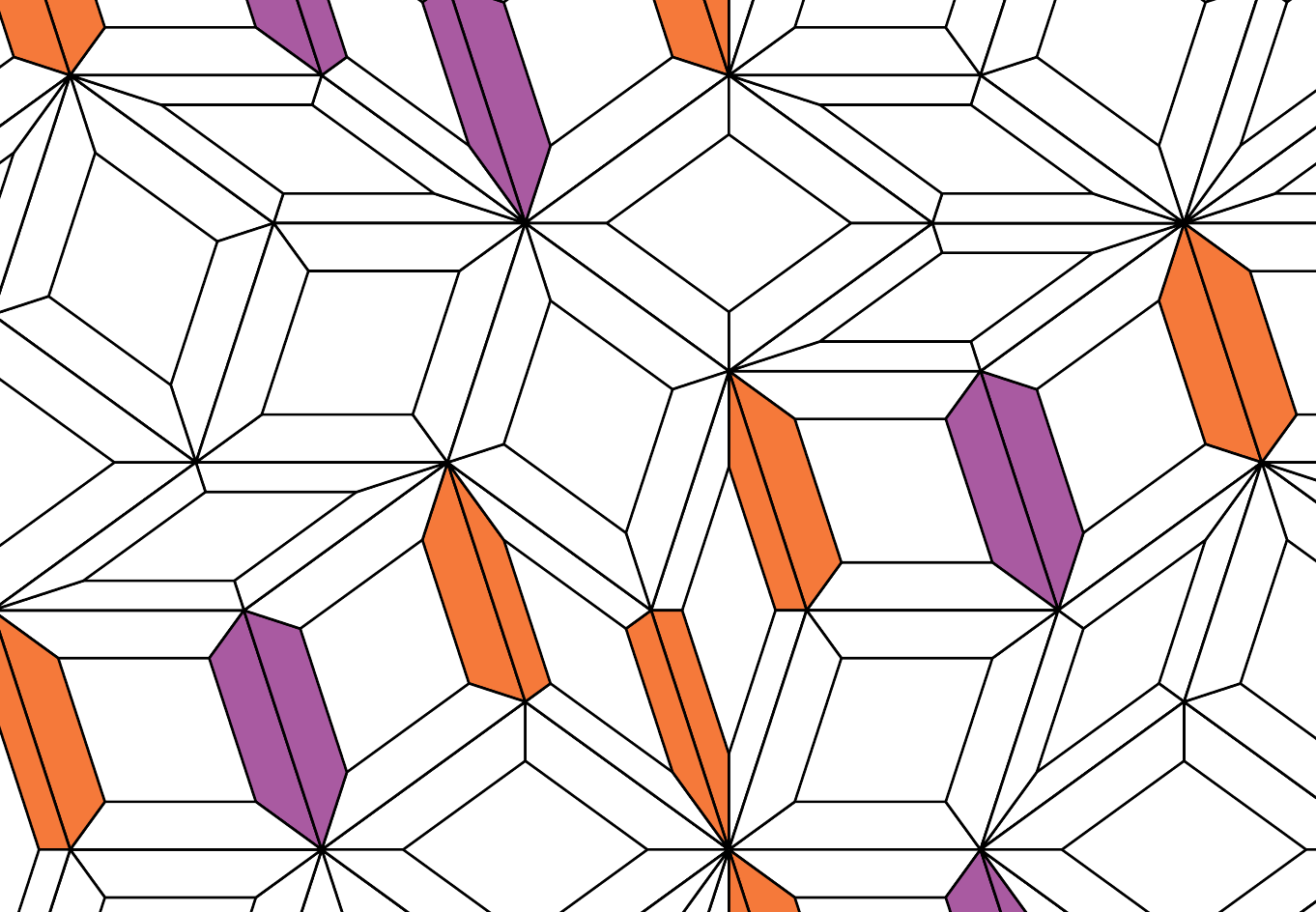}
		\caption{An example of valid patch.
		}
	\end{subfigure}
	\caption{Rhombus wang tiles.}
	\label{fig:wang_rhombus}
\end{figure}
A symbolic-geometrical subshift $X$ is given by a set of shapes (geometrical constraints), a set of forbidden patterns (geometrical subshift) and colourings on the tiles (symbolic constraints).
Even when the geometrical subshift is a full shift, geometrical and symbolic constraints can interact in interesting ways. For example, there is a choice of tiles on the Penrose rhombuses such that all valid tilings correspond to a geometrical Penrose tiling after erasing colours (in particular, no valid tiling use a single shape); see Appendix \ref{appendix:penrose}.

The definitions of minimal radius $r$ patterns and locally allowed patterns extend naturally to symbolic-geometrical tilings. 

\subsection{Computability and decidability}
\begin{definition}
	A \emph{decision problem} is a function $A : \dom(A) \to \{0,1\}$, where $\dom(A)$ is called the \emph{input domain} of $A$.
\end{definition}
\begin{definition}
	A decision problem $A$ is said to be decidable when there exists an algorithm (or Turing machine) that, given as input any $x\in dom(A)$, terminates and outputs $A(x)$.
\end{definition}
A weaker notion of computability for decision problems is the following:
\begin{definition}[co-computably enumerable, $\coRE$]
	A decision problem $A$ is called \emph{co-computably enumerable}, also known as \emph{co-recursively enumerable}, when there exists a total computable function $f : \dom(A)\times\mathbb{N} \to \{0,1\}$ such that:
	\[\forall x\in \dom(A), \quad A(x)\Leftrightarrow \forall n \in \mathbb{N}, f(x,n)\]
	Alternatively, a decision problem $A$ is co-computably enumerable if there is an algorithm that, on input $x$, terminates if and only if $A(x)$ is false.
	
	We denote by $\coRE$ the class of co-computably enumerable problems.
\end{definition}
$\coRE$ is a class of the arithmetical hierarchy; see \cite{kozen2006,monin2022}.
\begin{definition}[Many-one reductions]	
	Given two decision problems $A$ and $B$, we say that $A$ \emph{many-one reduces} to $B$, and write $A\manyone B$, 
	when there exists a total computable function $f : \dom(A)\to \dom(B)$ such that $A = B\circ f$. 
\end{definition}

\begin{figure}[htp]
	\center
	\includegraphics[width=0.3\textwidth]{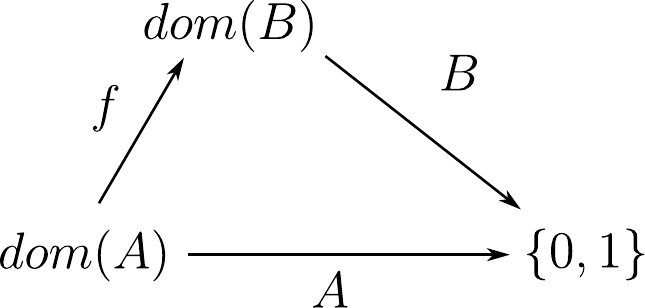}
	
	\caption{The diagram of a many one-reduction, $A\manyone B$ when $f$ is computable and the diagram commutes.}
	\label{fig:diagram_many-one}
\end{figure}

\begin{definition}[$\coRE$-hardness and $\coRE$-completeness]
	A problem $A$ is called $\coRE$-hard if $B \manyone A$ for any problem $B$ in $\coRE$.
	
	A problem $A$ is called $\coRE$-complete when it is both in $\coRE$ and $\coRE$-hard.
\end{definition}
Notice that $\coRE$-hard problems are undecidable. 
The canonical example of a $\coRE$-complete problem is the co-halting problem, that is the problem of deciding whether a Turing Machine does not terminate in finite time. 

Many-one reductions are a restrictive case of Turing reductions that are appropriate to study classes of decision problems such as $\coRE$, as the following Lemma shows: 

\begin{lemma}[$\coRE$-hardness]
	Given two problems $A$ and $B$ such that $A\manyone B$,
	\begin{itemize}
		\item if $B$ is $\coRE$, then $A$ is $\coRE$.
		\item if $A$ is $\coRE$-hard, then $B$ is $\coRE$-hard.
	\end{itemize}
\end{lemma}

\subsection{Domino problems}
In this paper, our goal is to prove the $\coRE$-hardness of a generalisation of the classical Domino problem, which is known to be $\coRE$-complete.

The classical Domino problem asks, given as input a finite set of Wang tiles, \emph{i.e.}, square tiles with a colour on each edge, whether there exists an infinite valid tiling with these tiles: see Fig.~\ref{fig:wang}.

\begin{theorem}[Berger66 \cite{berger1966}]
	The classical Domino problem is $\coRE$-complete.
\end{theorem}
Berger's paper provides a many-one reduction to the co-halting problem. 
We extend this classical problem to rhombus-shaped Wang tiles.
\begin{figure}[b]
	\begin{subfigure}{0.4\textwidth}
		\center
		\includegraphics[width=\textwidth]{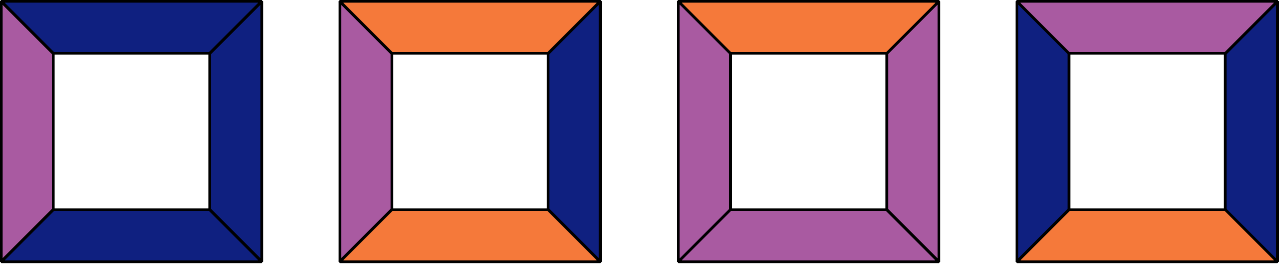}
		\caption{A tileset of 4 tiles.}
	\end{subfigure}\hfill
	\begin{subfigure}{0.55\textwidth}
		\center
		\includegraphics[width=0.8\textwidth]{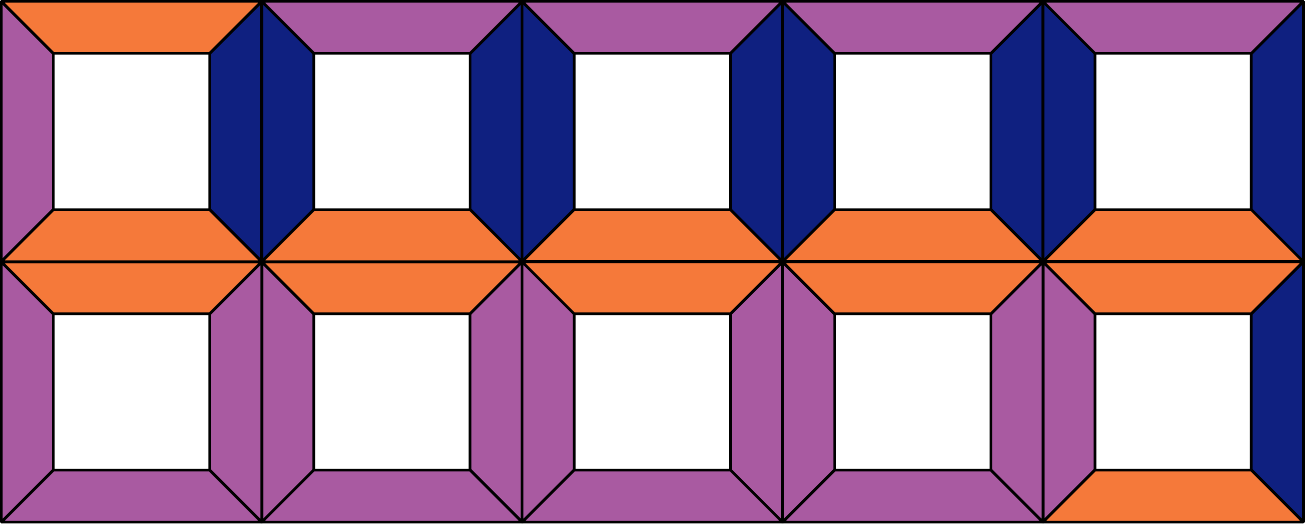}
		\caption{A valid patch with these tiles.}
	\end{subfigure}
	\caption{Tiles and tilings: a patch or tiling is called \emph{valid} when any two adjacent tiles have the same colour on their shared edge.}
	\label{fig:wang}
\end{figure}

\begin{definition}[$\domino{X}$] 
	Given a subshift $X$ on $\shapeset$, the \texttt{Domino} problem on $X$ is defined as:
	\begin{description}
		\item[Input] A finite set $\tileset$ of $\shapeset$-tiles
		\item[Output] Is there a $\tileset$-tiling $x \in X_{\tileset}$ such that $\pi(x) \in X$?
	\end{description}
\end{definition}
The classical \texttt{Domino} problem is $\domino{X_{\{\square\}}}$, that is, the domino problem on the full shift with a single shape (usually a square shape, but any single rhombus works). 
$\domino{Penrose}$ would be, given a set of tiles on Penrose rhombuses as in Fig.~\ref{fig:wang_rhombus} looking for a Penrose tiling with matching edges.

\section{Complexity of the Domino problem on rhombus subshifts}
\subsection{The Domino problem on effective rhombus subshifts is $\coRE$}

\begin{lemma}
	Let $\shapeset$ be a finite set of rhombus shapes.
	The following problem is computable:
	\begin{description}
		\item[Input] an integer $n$, a finite list of forbidden patterns $\mathcal F$, and a tileset $\tileset$ of $\shapeset$-tiles
		\item[Output] the list of all minimal radius $n$ patterns of $\tileset$-tiles that avoid all patterns in $\mathcal F$
	\end{description}
	\label{lemma:locally_allowed}
\end{lemma}
\begin{proof}
	The algorithm is as follows:
	\begin{enumerate}
		\item By combinatorial exploration, try all possibilities and list all $\tileset$-patterns with minimal radius $n$. 
		\item Eliminate all listed patterns $x$ such that the colour-erased pattern $\pi(x)$ contains some pattern in $\mathcal F$.
		\item Output the remaining patterns.
	\end{enumerate}
	In Point 1, remember that the set of edge-to-edge tilings on a fixed finite set of rhombus tiles have finite local complexity, so this process terminates in finite time.
\end{proof}

\begin{proposition}[$\domino{X} \in \coRE$]
	\label{prop:coRE}
	For any shapeset $\shapeset$ and any susbhift $X$ on $\shapeset$ defined by a computable enumeration of forbidden patterns $\mathcal{F}$, $\domino{X}$ is co-computably enumerable.
\end{proposition}
Note that, if $X$ is not effective, then $\domino{X}$ is $\coRE$ when provided with an enumeration of $\mathcal F$ as oracle.

This is essentially the same proof as for the classical Domino problem.
\begin{proof}
	The following problem, called disk-tiling-$X$, is decidable: 
	\begin{description}
		\item[Input] A finite set $\tileset$ of $\shapeset$-tiles and an integer $n$
		\item[Output] Is there a valid (finite) patch $x$ with tiles in $\tileset$ such that $\pi(x)$ is a rank $n$ locally allowed pattern of $X$, \emph{i.e.}, $\pi(x)\in \mathcal{A}_n(X)$? 
	\end{description}
	Simply compute the first $n$ forbidden patterns $\mathcal F_n$, which is possible because $X$ is effective, then apply Lemma \ref{lemma:locally_allowed} on $(n, \mathcal F_n, \tileset)$.
	For any input tileset $\tileset$, both the geometrical subshift $X$ and the symbolic full shift $X_\tileset$ have Finite Local Complexity so they are compact \cite{robinson2004}, and 
	\[ \domino{X}(\tileset) \Leftrightarrow \forall n \in \mathbb{N},\ \text{disk-tiling-}X(\tileset,n).\]
	
	Remark that $\text{disk-tiling-}X(\tileset,n) = 1$ when there exists a rank $n$ locally allowed pattern, that is, a pattern of minimal radius $n$ that avoids the first $n$ forbidden patterns $\mathcal{F}$ of $X$ with tileset $\tileset$.
	If $\forall n \in \mathbb{N},\ \text{disk-tiling-}X(\tileset,n)$, there exists a sequence $(p_n)_{n\in\mathbb{N}}$ with $\pi(p_n) \in \mathcal{A}_n(X)$. 
	Since the radius of the patches tends to infinity, by compacity there exists a limit tiling $x$ to which a subsequence converges. 
	Now remark that $\pi(x)\in X$ because it avoids all forbidden patterns in $\mathcal{F}$. Indeed, for any $k$, the $k$th forbidden pattern does not appear in any $\pi(p_n)$ for $n\geq k$, so it does not appear in $\pi(x)$.
	Since disk-tiling-$X$ is computable, we indeed have domino-$X \, \in \coRE$.\qedhere
\end{proof}
If $X$ is a full shift on some shapeset $\shapeset$, it is easy to see that the corresponding Domino problem is $\coRE$-hard by reduction to the classical version: given a finite set $\tileset$ of square Wang tiles, choose an arbitrary shape in $\shapeset$ and colour it like $\tileset$, and colour every other shape with four new different \emph{fresh} colours (so that any valid tiling may only use the first shape). In the rest of the paper, we extend this idea to work on an arbitrary subshift $X$.

\subsection{The Domino problem on shape uniformly recurrent subshifts is $\coRE$-hard}
The key concept is the concept of \emph{chains} of rhombuses \cite{kenyon1993} (also called \emph{ribbons} \cite{senechal1996}).
\begin{definition}[Chains of rhombuses]
	We call \emph{chain of rhombuses} a bi-infinite sequence of rhombuses that share an edge direction; see Figure \ref{fig:chain}. 
\end{definition}
A chain of rhombuses is characterized by its normal vector $\vec{v}$: the direction of the common edge. 
\begin{figure}[b]
	\center
	\includegraphics[width=0.6\textwidth]{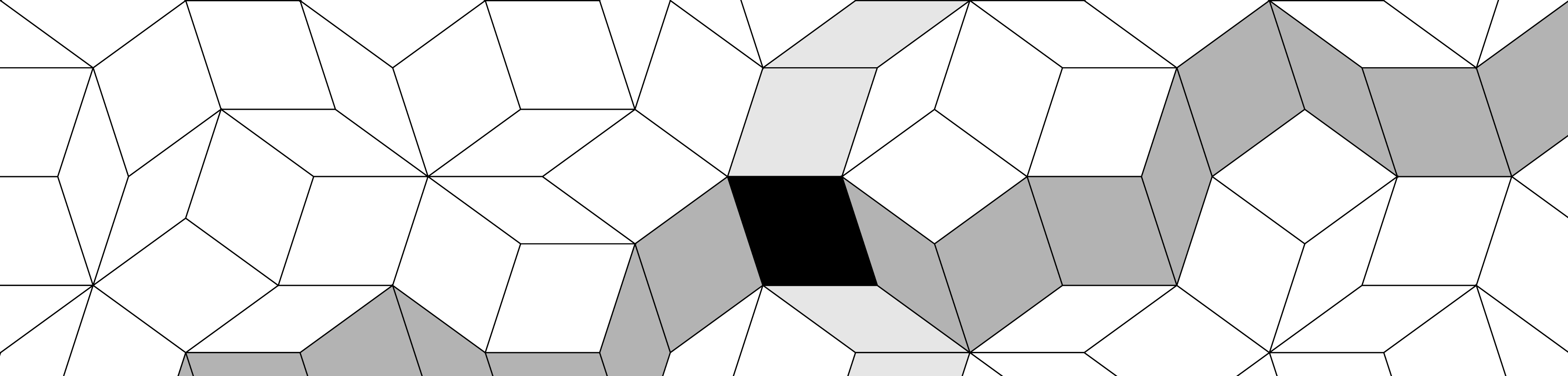}
	\caption{In an edge-to-edge rhombus tiling, a rhombus (in black) is at the intersection of two chains of rhombuses (in shades of grey).}
	\label{fig:chain}
\end{figure}

\begin{lemma}[Occurences of a rhombus, \cite{kenyon1993}]
	In an edge-to-edge rhombus tiling, rhombuses of edge directions $\vec{u}$ and $\vec{v}$ correspond exactly to the intersections of two chains of normal vectors $\vec{u}$ and $\vec{v}$. Moreover, two chains can cross at most once. See Fig.~\ref{fig:chain}.
	\label{lemma:chains} 
\end{lemma}
As a consequence, two chains of same normal vector cannot cross, otherwise there would be an impossible flat rhombus at the intersection. Such chains are therefore called parallel.

\begin{lemma}[Uniform monotonicity]
	Given a finite shapeset $\shapeset$, let $\theta_{min}$ be the smallest angle in a rhombus of $\shapeset$.
	For any $\shapeset$-tiling $x$,
	for any rhombus $r$ appearing in a chain $c$ of normal vector $\vec{u}$,
	the chain $c$ is outside the cone centered in $r$ and of half-angle $\theta_{min}$ along $\vec{u}$; see Fig.~\ref{fig:uniform_monotonicity}.
\end{lemma}
Overall, an edge-to-edge rhombus tiling can be decomposed as $d$ sets of parallel chains of rhombuses where $d$ is the number of edge directions.
Given an edge direction $\vec{u}$, the $\vec{u}$ chains can be indexed by either $\mathbb{Z}$, $\mathbb{N}$, $-\mathbb{N}$ or a finite integer interval in such a way that, starting from any position and moving along $\vec{u}$ one crosses the $\vec{u}$ chains in increasing order. 
\begin{figure}[t]
	\center\includegraphics[width=0.5\textwidth]{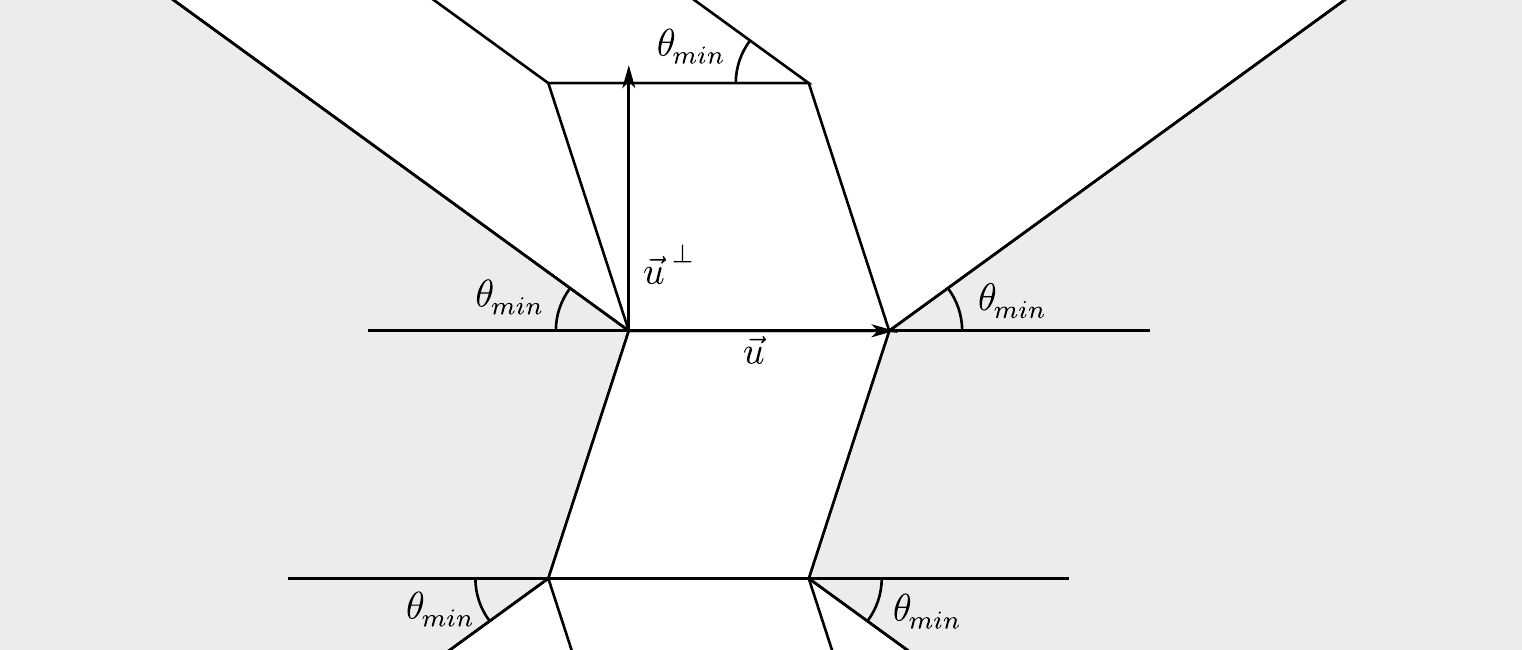}
	\caption{
	A rhombus $r$ and the chain $c$.
	The chain $c$ is outside the grey cones left and right of the rhombus $r$.}
	\label{fig:uniform_monotonicity}
\end{figure}

\begin{figure}[!b]
	\begin{subfigure}{0.55\textwidth}
		\center
		\includegraphics[width=\textwidth]{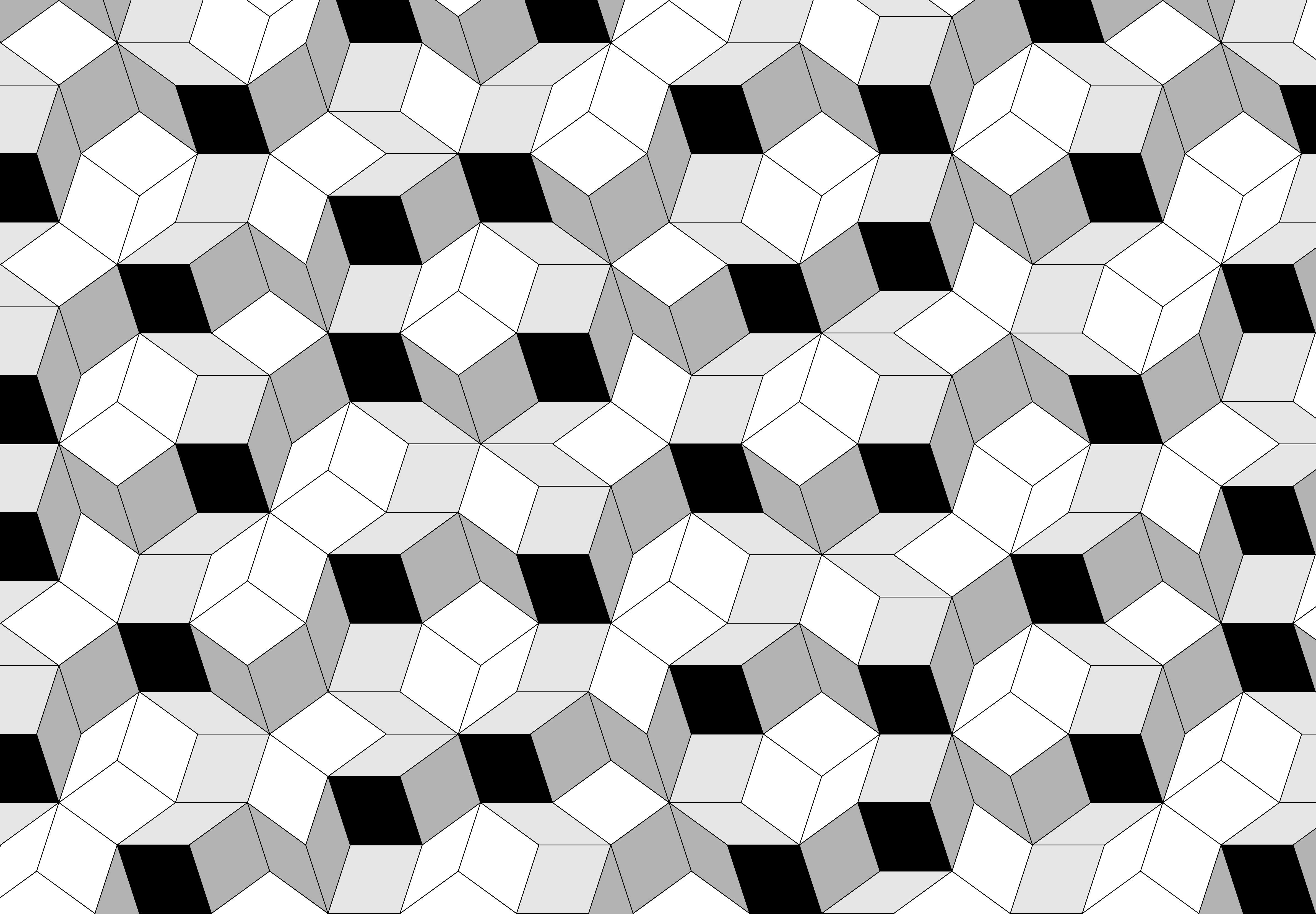}
		\caption{Occurences of $r$ (in black) in a Penrose tiling. 
		The chains that link the occurences of $r$ are highlighted in medium grey and light grey.
		}
		\label{fig:rhombus_ZZ2}
	\end{subfigure}
	\hfill
	\begin{subfigure}{0.4\textwidth}
		\center
		\includegraphics[width=\textwidth]{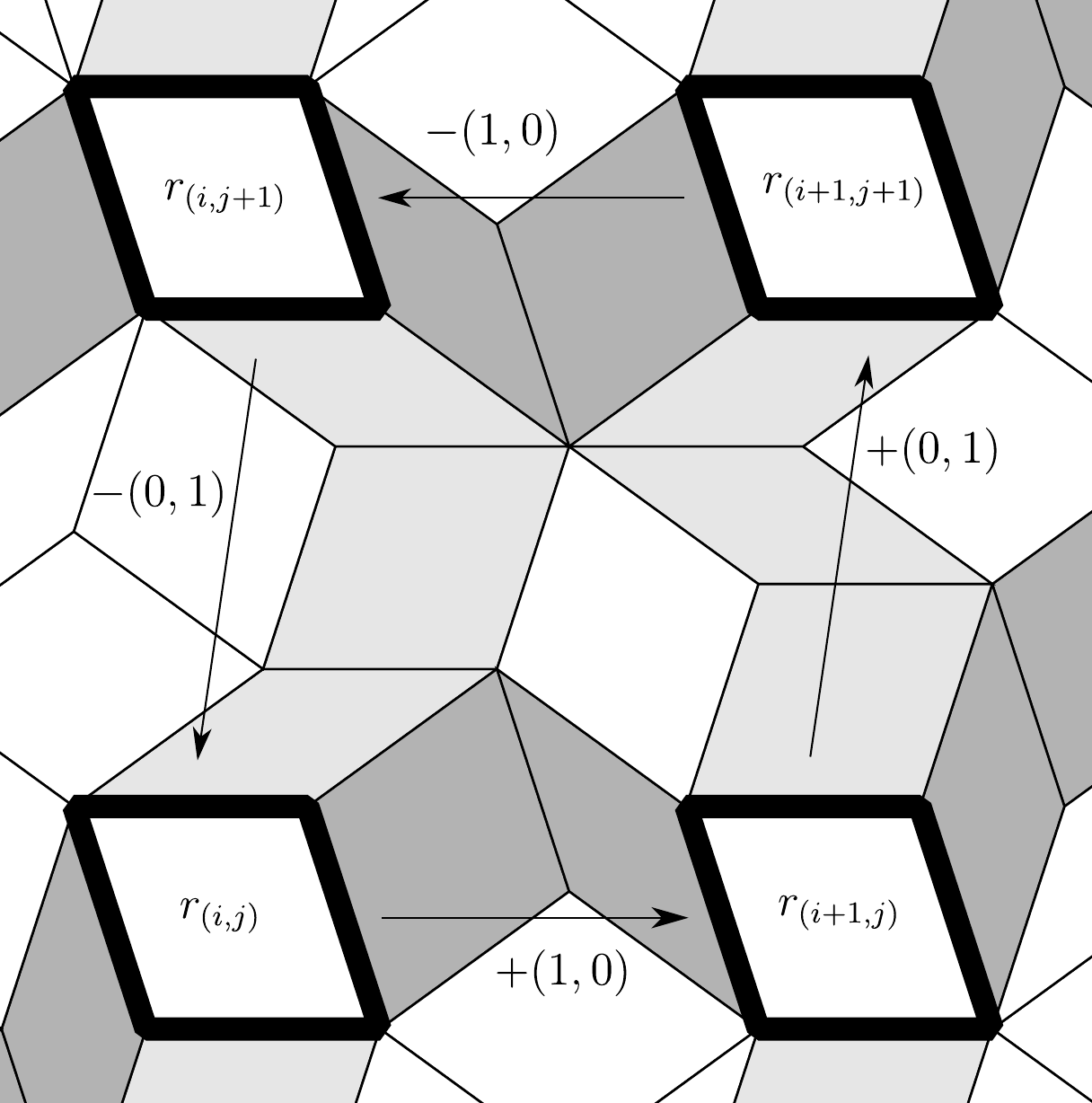}
		\caption{Indexing the occurences of $r$.}
		\label{fig:rhombus_ZZ2_details}
	\end{subfigure}
	\caption{Occurences of a uniformly recurrent rhombus $r$ in a tiling.}
\end{figure}

\begin{definition}[Shape uniform recurrence]
	Given a rhombus shape $r$ and a tiling $x$ we say that $r$ is uniformly recurrent in $x$, or that $x$ is $r$-uniformly recurrent, if $r$ appears in any disk of radius $R$ in $x$ for some $R$.
	
	A tiling $x$ is called shape uniformly recurrent when it is $r$-uniformly-recurrent for every shape $r$ that appears in $x$.
	
	A subshift $X$ is called shape uniformly recurrent when, for every shape $r$ that appears in some tiling $y \in X$, every tiling $x\in X$ is $r$-uniformly-recurrent.
\end{definition}
Note that this is much weaker than the usual uniform recurrence, which holds for every pattern instead of a single shape. 
\begin{lemma}
	Let $x$ be a edge-to-edge rhombus tiling and $r$ a shape that is uniformly recurrent in $x$.
	The occurences of $r$ in $x$ can be indexed by coordinates in $\mathbb{Z}^2$ such that two consecutive occurences of $r$ along a chain have adjacent $\mathbb{Z}^2$ coordinates; see Fig.~\ref{fig:rhombus_ZZ2_details}.
	\label{lemma:uniform}
\end{lemma}
\begin{proof}
	Let $x$ be a tiling in $X$ and $r$ a uniformly recurrent shape in $x$.
	Denote by $\vec{u}$ and $\vec{v}$ the two edge directions of the rhombus $r$. 
	As explained in Lemma \ref{lemma:chains}, the occurences of $r$ are exactly the intersections of a $\vec{u}$ chain and a $\vec{v}$ chain.
	
	As seen above, the $\vec{u}$ chains can be indexed by either $\mathbb{Z}$, $\mathbb{N}$, $-\mathbb{N}$ or a finite integer interval. Since $r$ is uniformly recurrent, only the case of $\mathbb{Z}$ is possible.
	Indeed by uniform recurrence of $r$, there exists $R$ such that any disk of radius $R$ in $x$ contains an occurence of $r$, and in particular intersects a $\vec{u}$ chain. So, starting from any position, one finds arbitrarily many $\vec{u}$ chains in both the $\vec{u}$ and $-\vec{u}$ directions.
	The same holds for $\vec{v}$.
	
	Denote $r_{(i,j)}$ the occurence of $r$ at the intersection of the $i$th $\vec{u}$ chain and the $j$th $\vec{v}$ chain.
	By definition of the indexing of chains, we see on Fig.~\ref{fig:rhombus_ZZ2_details} that starting from occurence $r_{(i,j)}$ and going along a $\vec{u}$ chain, the next occurence of $r$ is $r_{(i+1,j)}$.\qedhere
\end{proof}

\begin{proposition}
	Let $X$ be a non-empty subshift of edge-to-edge rhombus tilings that is shape uniformly recurrent.
	$\domino{X}$ is $\coRE$-hard.
	\label{prop:core-hard}
\end{proposition}
\begin{proof}
	We proceed by many-one reduction to the classical Domino problem which is known to be $\coRE$-complete \cite{berger1966}.
	\begin{figure}[b]
		\center
		\includegraphics[width=0.7\textwidth]{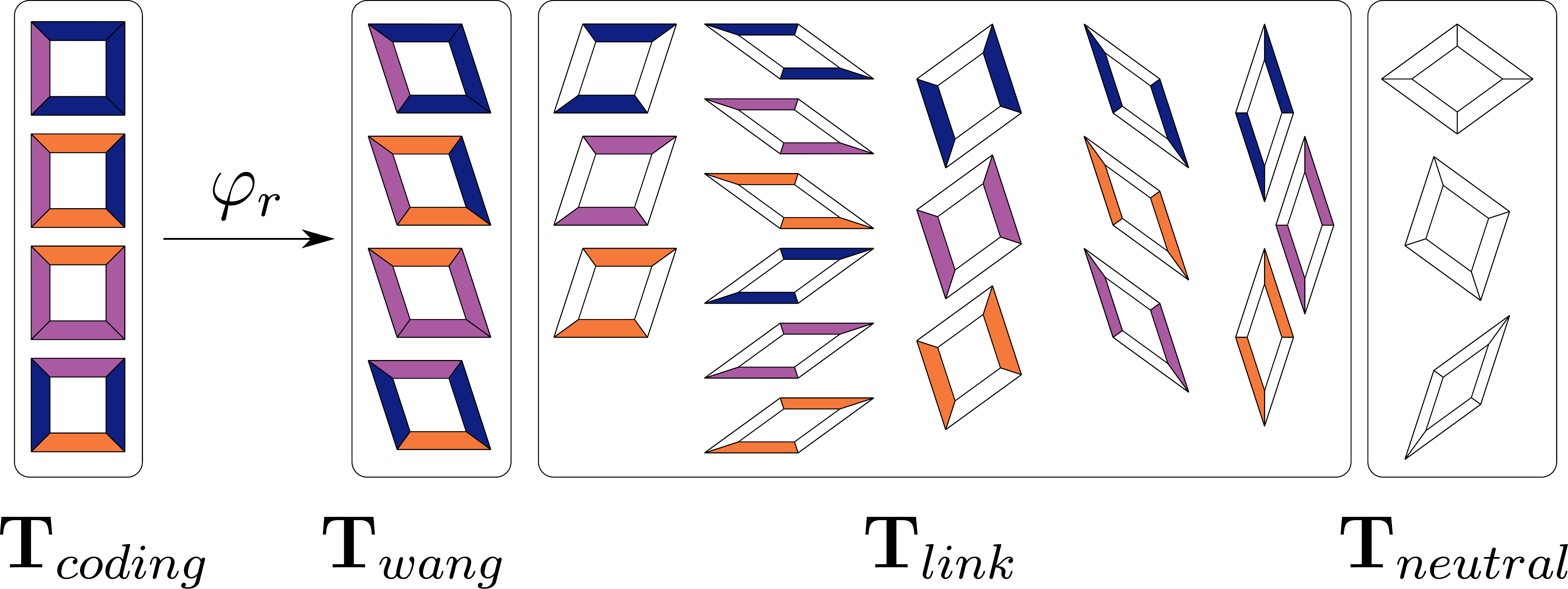}
		\caption{The reduction $\reduction_r$.}
		\label{fig:reduction_r}
	\end{figure}
	
	Let $\shapeset$ be the shapeset on which $X$ is defined, and choose an arbitrary shape $r\in \shapeset$. Define the reduction $\reduction_r$ as follows. We are given as input a finite set of square tiles $\tileset_{wang}$ on the set of coulours $S_{wang}$. 
	Define a tileset $\tileset_{rhombus} = \reduction_r(\tileset_{wang})$ on colours $S_{rhombus} := S_{wang} \cup \{ \blank \}$, where $\blank$ is a fresh colour, as $\tileset_{rhombus} := \tileset_{coding} \cup \tileset_{link} \cup \tileset_{neutral}$, with:
	\begin{enumerate}
		\item $\tileset_{coding}$ is a copy of $\tileset_{wang}$ on the shape $r$. Formally, for each tile $(a_0, a_1, a_2, a_3) \in \tileset_{wang}$, $\tileset_{coding}$ contains a tile $(r, a_0, a_1, a_2, a_3)$. See Fig.~\ref{fig:reduction_r}.
		\item $\tileset_{link}$ is a set of rhombus tiles that link occurences of $r$ and transmit the colours. Formally, for each shape $r' \neq r$ such that $r$ and $r'$ share exactly one edge direction, say $\vec{u}$, and for each colour $a\in S_{wang}$ 
		$\tileset_{link}$ contains a tile of shape $r'$ with colour $a$ on both edges along $\vec{u}$ edges and colour $\blank$ on other edges. See Fig.~\ref{fig:reduction_r}.
		\item $\tileset_{neutral}$ completes the tileset with $\blank$ colour. Formally, for each shape $r'$ that shares no edge direction with $r$, $\tileset_{neutral}$ contains one tile of shape $r'$ with $\blank$ colour on each edge. See Fig.~\ref{fig:reduction_r}.
	\end{enumerate}
	We prove that $\tileset_{wang}$ admits a valid tiling of $\mathbb{Z}^2$ if and only if $\reduction_r(\tileset_{wang})$ admits a valid tiling $x$ such that $\pi(x)\in X$. 
	
	\begin{figure}[!b]
		\center
		\includegraphics[width=0.7\textwidth]{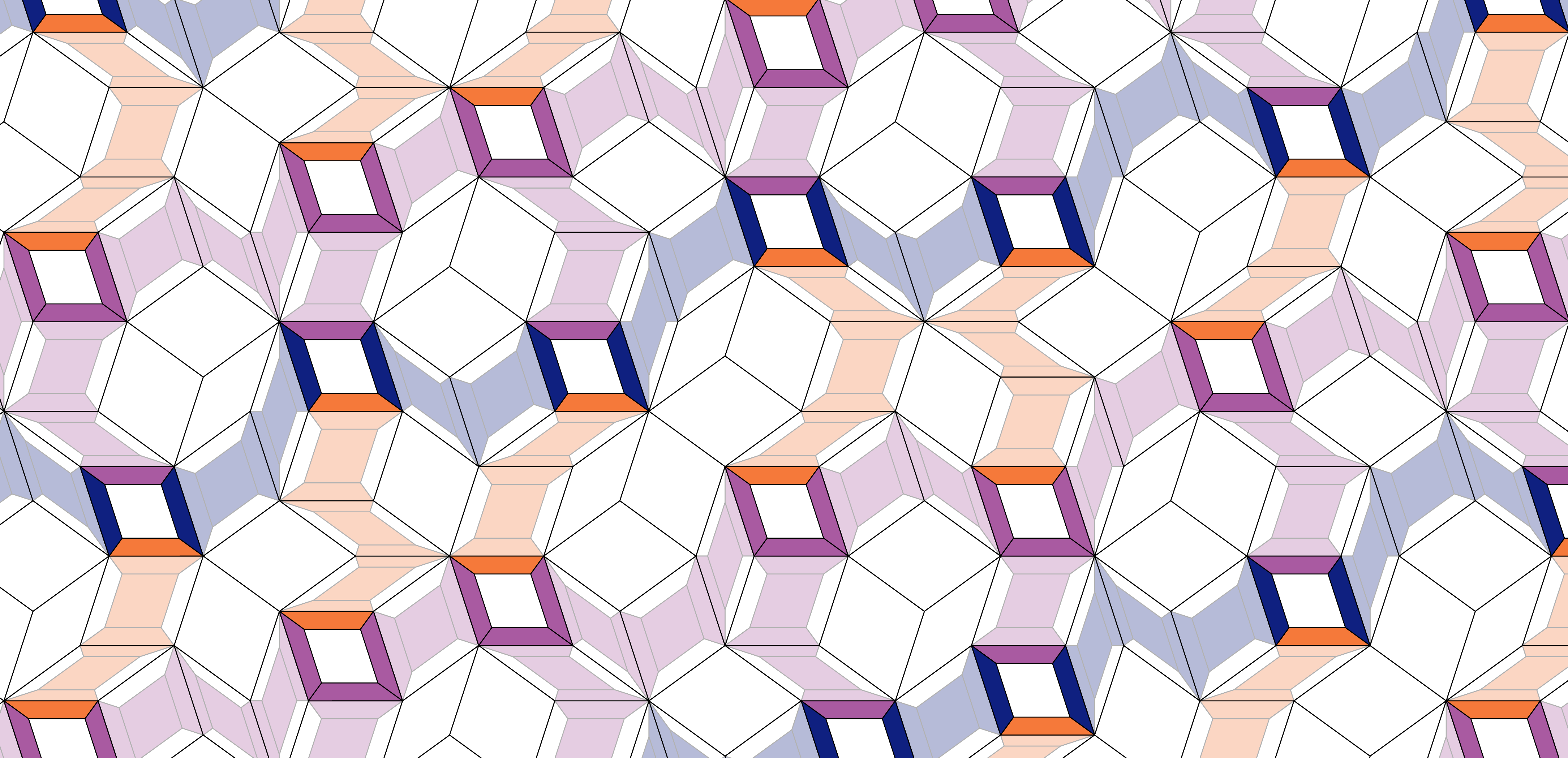}
		\caption{A valid tiling $x$ with tileset $\reduction_r(\tileset_{wang})$.
		Link tiles and neutral tiles have been visually modified to improve readability.}
		\label{fig:reduced_tiling}
	\end{figure}
	Assume that $\tileset_{wang}$ admits a valid $x_{wang}$ tiling of $\mathbb{Z}^2$.
	Let us pick $x\in X$.
	By hypothesis, $r$ is uniformly recurrent in $X$. We colour $x$ as follows (see Fig.~\ref{fig:reduced_tiling}):
	\begin{description}
		\item[Coding tiles:] index occurences of $r$ as $\{r_{(i,j)}| (i,j) \in \mathbb{Z}^2\}$ as explained in Lemma \ref{lemma:uniform}.
		Copy the tiles from $x_{wang}$ to the occurences of $r$, \emph{i.e.}, if at position $(i,j)$ in $x_{wang}$ there is a $(a_0,a_1,a_2,a_3)$ tile, then colour $r_{(i,j)}$ as the coding tile $(r_{(i,j)}, a_0, a_1, a_2, a_3)$.
		\item[Linker tiles:] by construction, the north colour of $r_{(i,j)}$ is equal to the south colour of $r_{(i,j+1)}$ so linkers of that colour can be put along that portion of chain, and similarly for east-west links. 
		\item[Neutral tiles:] remaining tiles share no edge direction with $r$, so they must be neutral tiles.
	\end{description}
	The converse also holds: if there exists a valid tiling $x$ on $\reduction_r(\tileset_{wang})$, since $r$ is uniformly recurrent its occurences can be indexed by $\mathbb{Z}^2$ (Lemma \ref{lemma:uniform}). By construction of the linker tiles, colours of the coding tiles $r_{(i,j)}$ correspond to a valid $\tileset_{wang}$ tiling in $\mathbb{Z}^2$. 
\end{proof}

\subsection{The Domino problem on any rhombus subshift is $\coRE$-hard}
In this section, for any subshift $X$, we build a subshift $X'$ that has a uniformly recurrent shape and such that $\domino{X'}\manyone \domino{X}$.

\begin{definition}[Restriction]
	Given any rhombus subshift $X$ on shapeset $\shapeset$, and a subset $\shapeset' \subset \shapeset$, define $\rho_{\shapeset'}(X)$ as the restriction of $X$ to the configurations that contain only shapes in $\shapeset'$, that is, $\rho_{\shapeset'}(X) := X \cap X_{\shapeset'}$.
\end{definition}

\begin{lemma}
	Let $X$ be a nonempty subshift on a shapeset $\shapeset$. 
	For any $r\in \shapeset$, either $r$ is uniformly recurrent in every $x\in X$, 
	or the restriction $\rho_{\shapeset\setminus\{r\}}(X)$ is nonempty.
	\label{lemma:reduction}
\end{lemma}

\begin{proof}
	This proof, once again, comes from finite local complexity and therefore compactness of edge-to-edge rhombus tilings.
	
	If $r$ is not uniformly recurrent in every $x \in X$, by definition there exist arbitrarily large patterns in $X$ that do not contain $r$. 
	By compactness there exists an infinite tiling in $X$ containing no $r$. 
	Hence $\rho_{\shapeset\setminus\{r\}}(X)$ is non-empty.
\end{proof}

\begin{lemma}
	Let $X$ be a nonempty subshift on a finite set of rhombuses $\shapeset$. There is a subset $\shapeset'\subseteq \shapeset$ such that $\rho_{\shapeset'}(X)\neq \emptyset$, and there exists $r \in \shapeset'$ which is uniformly recurrent in every configuration $x\in \rho_{\shapeset'}(X)$.
	\label{lemma:subsubshift}
\end{lemma}

\begin{proof}
	We prove this by induction on the number of shapes, \emph{i.e.}, on $\#\shapeset$.
	
	If $\shapeset = \{ r\}$, take $\shapeset' = \shapeset$ and $r$ is uniformly recurrent in any tiling $x\in X$.
	
	Assume the result holds for at most $n$ shapes and $\#\shapeset = n+1$. 
	Pick some $r\in \shapeset$: by Lemma \ref{lemma:reduction} either $r$ is uniformly recurrent in all tilings $x\in X$ or $\rho_{\shapeset\setminus\{r\}}(X)$ is non-empty. 
	In the first case, we conclude with $\shapeset' = \shapeset$.
	In the second case, we apply the induction hypothesis on $\rho_{\shapeset\setminus\{r\}}(X)$, obtaining some $\shapeset' \subseteq \shapeset\setminus\{r\}$ and a tile $r'\in\shapeset'$ such that $r'$ is uniformly recurrent in $\rho_{\shapeset'}(X)$.
\end{proof}

\begin{theorem}
	Let $X$ be a non-empty subshift of edge-to-edge rhombus tilings. 
	$\domino{X}$ is $\coRE$-hard.
\end{theorem}
This result, with Proposition~\ref{prop:coRE}, implies that $\domino{X}$ is $\coRE$-complete for effective subshifts $X$.

\begin{proof}
	By Lemma \ref{lemma:subsubshift} there exists a subset of the shapeset $\shapeset' \subset \shapeset$ such that the subshift $X' := \rho_{\shapeset'}(X)$ has a uniformly recurrent tile $t$.
	\begin{figure}[t]
		\center
		\includegraphics[width=0.6\textwidth]{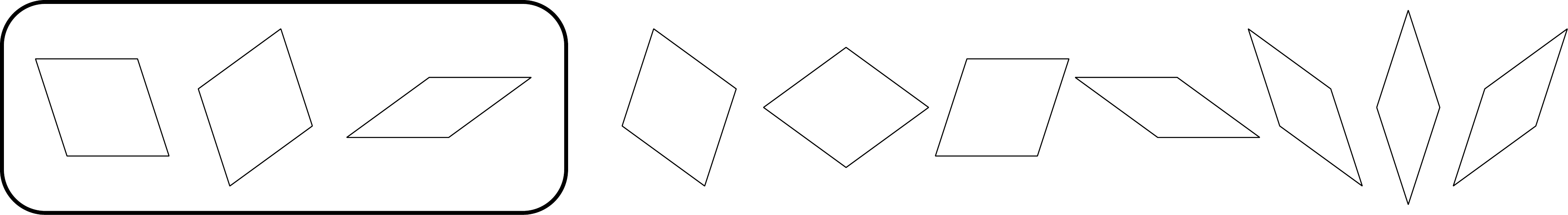}
		\caption{A set of shapes $\shapeset$, and in the left box the subset $\shapeset'$.}
		\label{fig:subshapeset}
	\end{figure}
	
	By Proposition \ref{prop:core-hard}, $\domino{X'}$ is $\coRE$-hard. 
	We now show that $\domino{X'} \manyone \domino{X}$ so that $\domino{X}$ is also $\coRE$-hard.
	
	The many-one reduction $\reduction_{\shapeset'}$ from $\domino{X'}$ to $\domino{X}$ is defined as follows: given a $\shapeset'$-wang tileset  $\tileset'$ we define $\reduction_{\shapeset'}(\tileset'):= \tileset' \cup \tileset_{fresh}$ where $\tileset_{fresh}$ contains a tile with a \emph{fresh} colour on each edge for each shape in $\shapeset \setminus \shapeset'$; see Figure \ref{fig:reduction_projection}. 
	
	\begin{figure}[b]
		\center
		\includegraphics[width=0.8\textwidth]{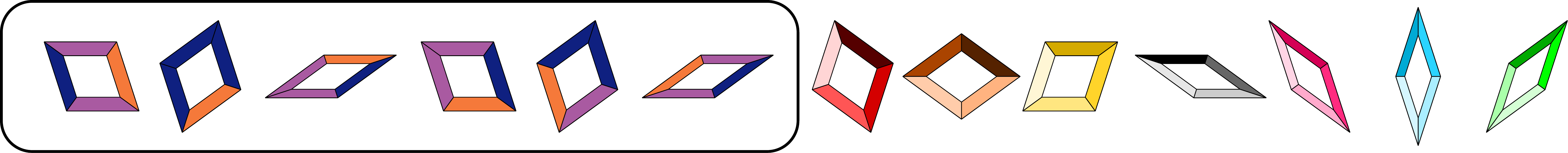}
		\caption{The reduction from a $\shapeset'$-tileset to a $\shapeset$-tileset : in the left box the original $\shapeset'$-tileset, on the right the tiles with fresh colours $\tileset_{fresh}$.}
		\label{fig:reduction_projection}
	\end{figure} 
	
	Remark that there is no reason that choosing the suitable $\shapeset'$ given $\shapeset$ and $X$ (as an enumeration of forbidden patterns) can be done in a computable manner.
	It only matters that, for a fixed $\shapeset'$, $\reduction_{\shapeset'}$ is computable, which is easily seen with the above definition.

	This reduction $\reduction_{\shapeset'}$ is well defined as we have 
	\[\domino{X'}(\tileset') \Leftrightarrow \domino{X}(\reduction_{\shapeset'}(\tileset')).\]
	
	The implication holds because $X'\subset X$ and $\tileset'\subset \reduction_{\shapeset'}(\tileset')$, so a $\tileset'$-tiling that projects by erasing colours in $X'$ is also a $\reduction_{\shapeset'}(\tileset')$-tiling that projects in $X$. 
	
	The converse holds because the tiles in $\tileset_{fresh}$ cannot appear in $\reduction_{\shapeset'}(\tileset')$-tilings because they have a fresh colour on each side so no tile can be placed next to it. 
	So a $\reduction_{\shapeset'}(\tileset')$-tiling $x$ is actually a $\tileset'$-tiling.
	Since $x$ contains only tiles in $\tileset'$, $\pi(x)$ contains only shapes in $\shapeset'$ so $\pi(x) \in \rho_{\shapeset'}(X) = X'$.
\end{proof}

\begin{remark}[Fresh colours]
	Remark that, if we remove the condition that a tileset on shapeset $\shapeset$ must contain at least a tile for each shape $r\in\shapeset$, we can take $\reduction(\tileset') := \tileset'$. In essence the fact of taking fresh colours simulates that.
\end{remark}

\begin{remark}[Restriction]	
	We could consider, instead of taking $X'$ as a restriction $\rho_{\shapeset'}(X)$, taking an arbitrary minimal subshift $X''\subset X$ (where all patterns are uniformly recurrent).
	However, there is no clear reduction from $\domino{X''}$ to $\domino{X}$.
\end{remark}

\begin{remark}[Beyond $\coRE$]
	If $X$ is not given by a computable enumeration of forbidden patterns $\mathcal F$, but $\mathcal F$ is given as an oracle, we remarked earlier that $\domino{X}$ is $\coRE$ relative to the oracle $\mathcal F$.
	However, we have only proved that $\domino{X}$ is $\coRE$-hard, but not relative to $\mathcal F$.
	In particular we have not proved that $\domino{X}$ is $\coRE$-complete relative to the oracle $\mathcal F$.
\end{remark}

\bibliography{dominiangelos}

\appendix
\section{Rhombus Wang tiles to define Penrose tilings}
\label{appendix:penrose}
Rhombus Penrose tilings were originally defined as jigsaw-type tiles with indentations on their edges \cite{penrose1974} which then reformulated as rhombus tiles with arrows as labels on their edges \cite{debruijn1981}.
With the arrow labels, two adjacent tiles must have the same type and direction of arrow on their shared edge. In both the original and modern definitions, the tiles are defined up to isometry.

These arrow labels are not strictly speaking the same as puting a single colour on each edge as for Wang tiles. However, if we define the tileset up to translation and not up to isometry we can translate these arrow labels as single colours, see Fig.~\ref{fig:arrow_to_colour}.

\begin{figure}[htp]
	\center
	\includegraphics[width=0.8\textwidth]{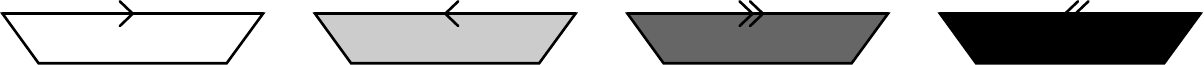}
	\caption{Given an edge direction, each type of arrow is translated into two colours.\\
	By rotating this figure by multiples of $\tfrac{2\pi}{5}$, one otbains the arrow-colour translation for all directions. }
	\label{fig:arrow_to_colour}
\end{figure}

This process gives us a tileset shown in Fig.~\ref{fig:penrose_rules_wang} that defines the subshift of Penrose tilings.

\begin{figure}[htp]
	\includegraphics[width=\textwidth]{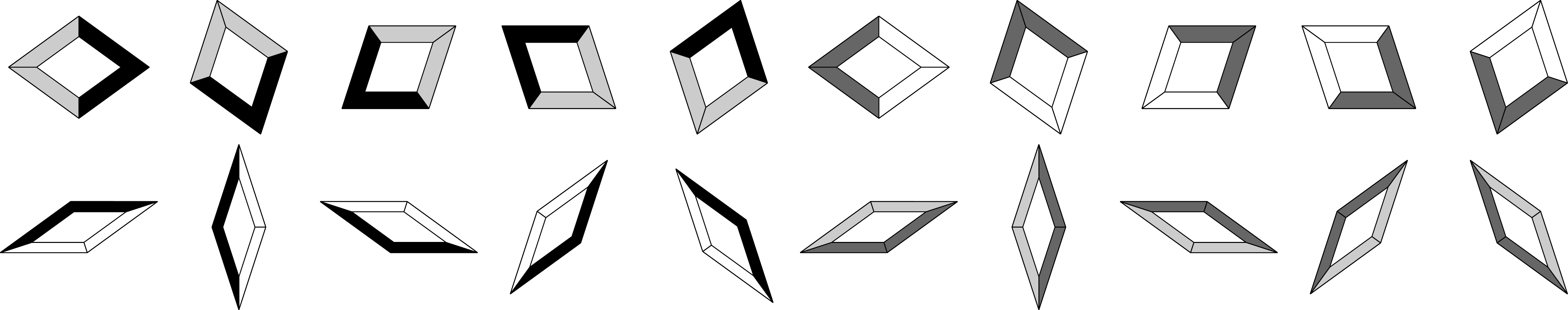}
	
	\caption{A rhombus Wang tileset definition of Penrose tilings.}
	\label{fig:penrose_rules_wang}
\end{figure}

This tileset contains 20 tiles up to translation, which can be considered as less elegant as the 2 tiles up to isometry of the classical definition. Note however that if we use Wang-type colours on the edges, there exists no aperiodic up-to-isometry tileset, indeed a single tile up-to-isometry always tiles the plane in a periodic way (see Fig.~\ref{fig:up_to_isometry}). This implies that if we want to define Penrose tilings with Wang rhombus tiles, we cannot use and up-to-isometry or up-to-rotation tileset.

\begin{figure}[htp]
	\center
	\includegraphics[width=0.8\textwidth]{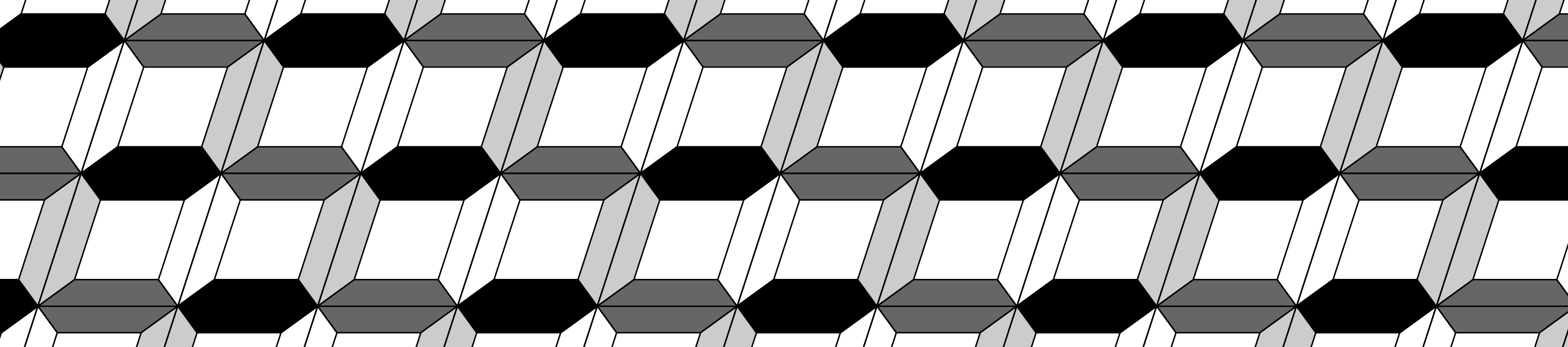}
	\caption{Given a rhombus Wang tile with different colours on all four edges, there exists a periodic tiling of the plane with only this tile up to translation and rotation (of angle $\pi$).}
	\label{fig:up_to_isometry}
\end{figure}

However, if we consider the Penrose tileset of Fig.~\ref{fig:penrose_rules_wang} up to $\tfrac{2\pi}{5}$ rotations we have a reduced tileset with 4 tiles shown in Figure \ref{fig:penrose_rules_wang_2pi5}. 
Note that, in this definition, tiles can only be rotated by multiples of $\tfrac{2\pi}{5}$ so the counterexample of Fig.~\ref{fig:up_to_isometry} does not apply as a rotation of angle $\pi$ is necessary for this counterexample.

\begin{figure}[htp]
	\center
	\includegraphics[width=0.5\textwidth]{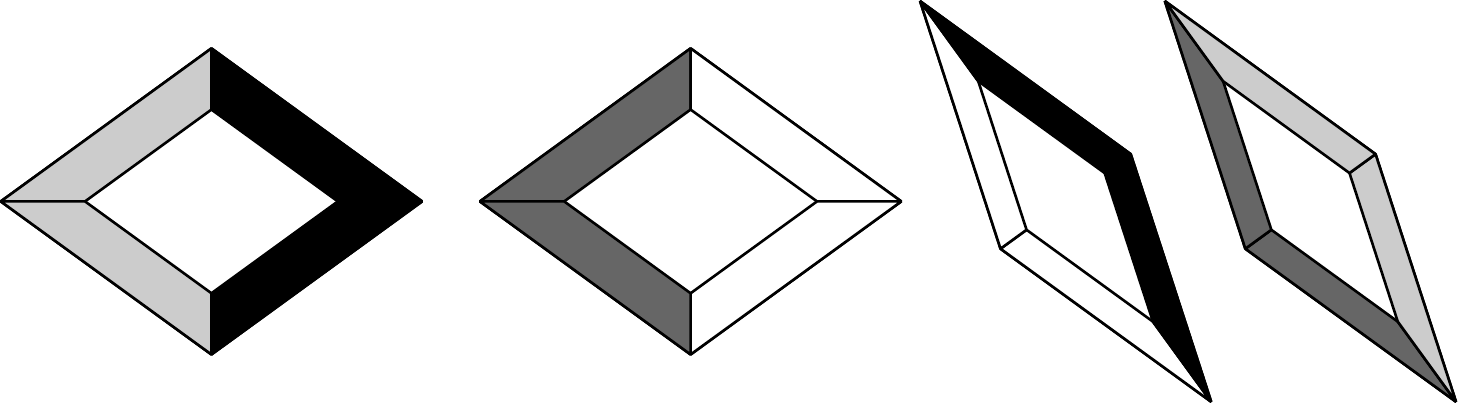}
	\caption{A rhombus Wang tileset up to translations and $\tfrac{2\pi}{5}$ rotations that defines Penrose tilings.}
	\label{fig:penrose_rules_wang_2pi5}
\end{figure}

\end{document}